\documentclass[%
preprint,
amsmath,
amssymb,
 aps,
prl,
]{revtex4-2}

\usepackage{setspace}

\usepackage[LGR,T1]{fontenc} % spell out all text encodings used
\usepackage[utf8]{inputenc} %
\usepackage{substitutefont} % so we can use fonts other than those in babel
\usepackage[polutonikogreek,english]{babel}
\usepackage[largesc]{newtxtext} %
\usepackage[varqu,varl]{zi4}% inconsolata
\usepackage{cabin}% sans serif
\usepackage[vvarbb,upint,slantedGreek,smallerops]{newtxmath}
\substitutefont{LGR}{\rmdefault}{Tempora} % use Tempora to render Greek text

\usepackage{textcomp}

\usepackage[margin=1.5in]{geometry}
\usepackage{hyperref}
\usepackage[hyphens]{xurl}
\usepackage{amsmath}
\usepackage{amsfonts}
\usepackage{amssymb}

\usepackage[pdftex]{graphicx}
\usepackage{lastpage}
\usepackage{slashed}

\usepackage{xcolor}

\usepackage{xfrac}

\usepackage{comment}
\usepackage[yyyymmdd,hhmmss]{datetime} % must be after babel

\usepackage{lastpage}
\usepackage[running,mathlines]{lineno}
\usepackage{isomath}
\usepackage{siunitx}
\usepackage{upgreek}

\usepackage{marginnote}

\usepackage{amsthm}

\newtheorem{theorem}{Theorem}

\newtheorem*{law*}{Law}

\theoremstyle{definition}

\usepackage{verbatim}

\usepackage{xspace}
\usepackage{braket}
\usepackage{mathtools}
\mathtoolsset{showonlyrefs=true}

\usepackage{siunitx}
\usepackage{physics}

\newcommand{\pr}{\ensuremath{\mathrm{P}}}

\newcommand{\piu}%
{\textrm{\greektext p}}
\newcommand{\eu}%
{\ensuremath{\mathrm{e}}}
\newcommand{\iu}%
{\ensuremath{\mathrm{i}}}

\providecommand{\newoperator}[3]{%
\newcommand*{#1}{\mathop{#2}#3}}

\newcommand{\tran}%
{\textsf{T}}
\newcommand{\herm}%
{\textsf{H}}
\newcommand{\deltau}%
{\textrm{\greektext d}}
\newcommand{\Deltau}%
{\textrm{\greektext D}}

\makeatletter
\providecommand*{\diff}%
{\@ifnextchar^{\DIfF}{\DIfF^{}}}
\def\DIfF^#1{%
\mathop{\mathrm{\mathstrut d}}
\nolimits^{#1}\gobblespace}
\def\gobblespace{%
\futurelet\diffarg\opspace}
\def\opspace{%
\let\DiffSpace\!%
\ifx\diffarg(%
\let\DiffSpace\relax
\else
\ifx\diffarg[%
\let\DiffSpace\relax
\else
\ifx\diffarg\{%
\let\DiffSpace\relax
\fi\fi\fi\DiffSpace}

\newoperator{\loglike}%
{\mathrm{loglike}}{\nolimits}

\newoperator{\notloglike}%
{\mathrm{notloglike}}{\limits}

\makeatother

\newcommand{\Schroedinger}{Schr\"{o}dinger\xspace}

\newcommand{\normalx}[3]{\ensuremath{\mathscr{N}\left(#1 \mid #2,#3\right)}\xspace}

\renewcommand{\piu}{\uppi}
\renewcommand{\deltau}{\updelta}

\let\originalpartial\partial
\let\partial\relax
\newrobustcmd*{\partial}{\text{\rotatebox[origin=t]{10}{\scalebox{0.95}[1]{\ensuremath{\originalpartial}}}}\hspace{-0.05em}}

\usepackage{graphicx}% Include figure files
\usepackage{dcolumn}% Align table columns on decimal point

\newcommand{\CC}{\ensuremath{C}\xspace}

\newcommand{\ee}{\ensuremath{\text{e}}\xspace}

\newcommand{\planckbar}{\ensuremath{\hbar}\xspace}

\newcommand{\hessian}{\ensuremath{\mathbf{H}}\xspace}

\usepackage[normalem]{ulem}

\usepackage{bm}

\newcommand{\Identitymatrix}{\ensuremath{\mathbf{I}}\xspace}

\newcommand{\Sigmabold}{\ensuremath{\bm{\mathrm{\Sigma}}}\xspace}

\newcommand{\kernBeforeIntegral}{\ensuremath{\kern-0.17em}\xspace}

\newcommand{\mechanics}{mechanics\xspace}

\newcommand{\eqdefA}{=\xspace}

\newcommand{\entropyS}{\ensuremath{\mathrm{S}\xspace}}

\newcommand{\QFT}{\ensuremath{\text{QFT}}\xspace}

\usepackage[mathscr]{euscript}

\usepackage[en-US]{datetime2}

\usepackage[inline]{enumitem}

\newcommand{\state}{\ensuremath{\text{state}}\xspace}

\usepackage{changepage}
\begin{document}

\thispagestyle{empty}
\newpage

\title{Quantum Entropy}
\author{{Davi Geiger} and {Zvi M.\ Kedem}}
\affiliation{ Courant Institute of Mathematical Sciences\\
 New York University, New York, New York 10012}

\begin{abstract}

Quantum physics, despite its observables being intrinsically of a  probabilistic nature, does not have a quantum entropy assigned to them.  We propose a quantum entropy that quantify the randomness of a pure quantum state via a conjugate pair of observables forming the quantum phase space. The entropy is dimensionless, it is a relativistic scalar, it is invariant under coordinate transformation of position and momentum that maintain conjugate properties, and under CPT transformations; and its minimum is positive due to the uncertainty principle.

We expand the entropy to also include mixed states and show that the proposed entropy is always larger than von Neumann's entropy. 

We conjecture an entropy law whereby that entropy of a closed system never decreases, implying  a  time arrow for particles physics.

\end{abstract}

\maketitle

\pagebreak

\tableofcontents

\pagebreak

\section{Introduction}

The concept of entropy has been useful in classical physics but extending it to quantum \mechanics (QM)  has been challenging. In classical physics Boltzmann entropy and Gibbs entropy and their respective H-theorems~\cite{gibbs2014elementary} are formulated in the classical phase space, capturing the practical limitations  (described as a randomness) of specifying the degrees of freedom (DOFs) of a system (a classical state). Naturally, in quantum physics the  DOFs specify a quantum state. Von Neumann entropy, in analogy to classical physics,  quantify the randomness of specifying the quantum state, expressed by the classical statistical  coefficients of a mixture of  quantum states.  All pure states are assigned zero von Neumann entropy. 

Our starting point are pure states, where the DOFs are precisely specified. We investigate the inherent quantum randomness associated with the physical observables (the eigenvalues of   Hermitian observable operators). This randomness is fully captured by conjugate pairs of observables  satisfying the uncertainty principle \cite{Robertson1929},  one pair being space and momentum and  the other pair associated with the internal  spin state. 

Our approach to define a quantum entropy is to quantify  both  (i) the inherit randomness of the observables and (ii) the randomness due to the limitations of specifying the DOFs of the quantum state. 

We point out that a quantum entropy is not an observable as there is no entropy operator, instead,  entropy is a scalar function associated with a state. Thus, we also require the quantum entropy to be a scalar invariant under special relativity, point wise transformations of coordinates, and CPT transformations.
 
We propose an entropy defined in quantum phase spaces  that satisfies those conditions. Here we concentrate in the quantum coordinate phase space,  the space of all possible states projected simultaneously in the position and a spatial frequency (momentum) basis. The entropy is applicable to both QM and Quantum Field Theory (QFT).  We define  the spin entropy elsewhere\cite{geiger2021spin}.

\subsection{Previous work}

Von Neumann entropy \cite{von2013mathematische} captures the  randomness associated with not-knowing precisely the quantum state, but does not capture the randomness associated with the observables. Thus, it requires the existence of classical statistics elements (mixed states) in order not to vanish.    Wehrl entropy \cite{wehrl1978general} is based on Husimi's \cite{husimi1940some} quasiprobability distribution, rooted in projecting states to an overcomplete basis representation of coherent states, which are not relativistic invariant. Note that no two coherent states are orthogonal to each other. Therefore, Kolmogorov third axiom for a probability distribution, for elementary events to be mutually exclusive,  is not satisfied. Consequently some probability properties, such as  the monotonicity of probabilities and the complement rule,  are not satisfied by Husimi's quasiprobability distribution. These limitations prevent Wehrl entropy from correctly counting the random values of the observables.  For example,  a quantum state where the projection in position space is a Dirac delta function at $\mathbf{r}_0$ produces  a non-zero value distribution for all possible position coordinates $\mathbf{r}$ in the  classical phase space coordinate $(\mathbf{r},\mathbf{p})$, where $\mathbf{p}$ is the momentum coordinate. Clearly, this  is not the  description of the random position  $\mathbf{r}$ in QM. Indeed Wehrl \cite{wehrl1977} referred to his proposed entropy as a classical entropy for the classical phase space.   

\subsection{Scenarios examined and a new physical law}
We examine some physical scenarios to evaluate the entropy evolution, including a coherent state evolution through  a potential free Dirac equation, and  the hydrogen atom in an excited state transitioning to the ground state with a photon emission. In  these scenarios the entropy increases. Clearly, in these two scenarios,  starting the solution in the identical final state, applying  time reversal,   and running the unitary transformations backwards,  would lead to entropy decreases. In \cite{GeigerKedem2021c} we studied a collision of two spinless particles and as they come close  with each other, due to the superposition of the position wave functions, the entropy can decrease.

We then  conjecture that there is an entropy law, universally applicable to  particle physics, stating that in a closed physical scenario the entropy never decreases. The motivation is that the  inverse of the amount of randomness (information) can not be gained in a closed system.  Such law implies irreversibility of time for all physical scenarios where entropy does not stay constant.  We will examine the consequences of such a law in physics.

\section{Quantum Entropy in Quantum Phase Spaces}
\label{sec:quantum-entropy-def}

We now proceed to define a  quantum entropy. It must account for both the coordinate  and the internal  (spin) DOFs. It should quantify  all the randomness associated with the observables of a quantum state.  Thus, the probability distributions that define the quantum entropy should express the uncertainty relations of the observables. 
Also when considering quantum mixed states, the quantum entropy should also quantify  all the classical randomness associated with the specification of a quantum state. Moreover, a quantum entropy  should  be  invariant under special relativity transformations,  under point-wise transformation of position coordinates (such as a spherical-polar coordinate change), and under CPT transformations.

\subsection{Coordinate-Entropy}
\label{sec:coordinate-DOF}
We associate with a state $\ket{\psi}$ its projection onto the QM eigenstates of the operators $\hat{\mathbf {r}}$ and $\hat{\mathbf {p}}$ , i.e., $\ket{\mathbf {r}}$ and $\ket{\mathbf {p}}$.
Either one projection,  $\psi(\mathbf{r})=\bra{\mathbf{r}}\ket{\psi}$ or  $\phi(\mathbf{p})=\bra{\mathbf{p}}\ket{\psi}$, is sufficient to recover the other one via a Fourier transform. However, to account for the randomness of the observables, both quantities are needed. In order to define the quantum phase space, we will use the density operator  $\rho=\ket{\psi}\bra{\psi}$. Projecting the density operator into the basis $\ket{\mathbf {r}}$ and $\ket{\mathbf {p}}$ we obtain the probability densities $\rho_{\mathrm{r}}(\mathbf{r})=\bra{\mathbf{r}}\rho\ket{\mathbf{r}}=|\psi(\mathbf{r})|^2$ and $\rho_p(\mathbf{p})=\bra{\mathbf{p}}\rho\ket{\mathbf{p}}=|\tilde \phi(\mathbf{p})|^2$.

We define the quantum coordinate phase space as the space of all possible states projected simultaneously on  pairs $\left ( \rho_{\mathrm{r}}(\mathbf{r}),\, \rho_p(\mathbf{p})\right)$.

A time evolution  according to a Hermitian Hamiltonian $H$ is described by $\rho_t=\eu^{-\iu \frac{H}{\hbar} t}\, \rho \, \eu^{\iu \frac{H}{\hbar} t}$ and so in quantum coordinate phase space is given by the pair of probability densities  $\rho_{\mathrm{r}}(\mathbf{r},t)=\bra{\mathbf{r}}\rho_t\ket{\mathbf{r}}=|\psi(\mathbf{r},t)|^2$ and $\rho_p(\mathbf{p},t)=\bra{\mathbf{p}}\rho_t\ket{\mathbf{p}}=|\tilde \phi(\mathbf{p},t)|^2$. 

We claim that is necessary to consider a quantum coordinate phase space in order to capture all the randomness of the observables, and to produce an entropy invariance  under applicable transformation.

Our formulation of the entropy is  motivated by previous work including  Boltzmann and Gibbs~\cite{gibbs2014elementary}, Shanon~\cite{shannon1948mathematical}, Jaynes~\cite{jaynes1965gibbs} and von Neumann\cite{von2013mathematische}. Let the entropy associated only  with the spatial coordinates be the differential entropy  $\entropyS_{\mathrm{r}}=  -\kernBeforeIntegral\int  \rho_{\mathrm{r}} (\mathbf{r},t) \ln \rho_{\mathrm{r}} (\mathbf{r},t) \, \diff^3\mathbf{r}$.  Let
$\mathbf{k}= \frac{1}{\hbar} \mathbf{p}$ be the spatial frequency (Fourier conjugate of $\mathbf{r}$),    $ \rho_k ( \mathbf{k},t)= \frac{1}{\hbar^3} \rho_p ( \mathbf{p},t)$  the associated probability density, and  $\entropyS_{\mathrm{k}}=-\kernBeforeIntegral\int  \rho_{\mathrm{k}} (\mathbf{k},t) \ln \rho_{\mathrm{k}} (\mathbf{k},t) \, \diff^3\mathbf{k}$. Then we define the entropy associated with the quantum coordinate phase space distributions as
\begin{align}
  \entropyS
=-\kernBeforeIntegral\int  \, \rho_{\mathrm{r}} (\mathbf{r},t)  \rho_k ( \mathbf{k},t)\,
 \, \ln \left ( \rho_{\mathrm{r}} (\mathbf{r},t)  \rho_k (\mathbf{k},t) \,  \right) \,
  \diff^3\mathbf{r}\, \diff^3\mathbf{k}
       = \entropyS_{\mathrm{r}}+ \entropyS_{\mathrm{k}}
\,.
\label{eq:differential-entropy}
\end{align}
The entropy is dimensionless and thus, invariant under changes of the units of measurements.

A natural extension of this  entropy to an $N$-particle QM system is
 \begin{align}
     S &= - \kernBeforeIntegral\int \diff^3 \mathbf {r}_1 \diff^3\mathbf {k}_1 \hdots \diff^3 \mathbf {r}_N \diff^3\mathbf {k}_N \,   \rho_{\mathrm{r}}(\mathbf {r}_1,\hdots, \mathbf {r}_N ,t)  \rho_{\mathrm{k}}(\mathbf {k}_1,\hdots, \mathbf {k}_N ,t)
     \\
     & \qquad \quad \times     \ln \left (  \rho_{\mathrm{r}}(\mathbf {r}_1,\hdots, \mathbf {r}_N ,t) \rho_{\mathrm{k}}(\mathbf {k}_1,\hdots, \mathbf {k}_N ,t) \right)
     \\
     &= -\kernBeforeIntegral\int \diff^3 \mathbf {r}_1 \hdots  \kernBeforeIntegral\int \diff^3 \mathbf {r}_N \,   \rho_{\mathrm{r}}(\mathbf {r}_1,\hdots, \mathbf {r}_N ,t)
     \ln \rho_{\mathrm{r}}(\mathbf {r}_1,\hdots, \mathbf {r}_N ,t)
     \\
     & \quad -  \kernBeforeIntegral\int \diff^3 \mathbf {k}_1 \hdots  \kernBeforeIntegral\int \diff^3 \mathbf {k}_N\,   \rho_{\mathrm{k}}(\mathbf {p}_1,\hdots, \mathbf {k}_N ,t)
     \ln \rho_{\mathrm{k}}(\mathbf {k}_1,\hdots, \mathbf {k}_N ,t) \, ,
     \label{eq:entropy-many-particles}
 \end{align}
 where $\rho_{\mathrm{r}}(\mathbf {r}_1,\hdots, \mathbf {r}_N ,t)=|\psi(\mathbf {r}_1,\hdots, \mathbf {r}_N ,t)|^2$ and $\rho_{\mathrm{k}}(\mathbf {k}_1,\hdots, \mathbf {k}_N ,t)=|\phi(\mathbf {k}_1,\hdots, \mathbf {k}_N ,t)|^2 $ are defined in QM via the projection of the state $\ket{\psi_t}^N$ of $N$ particles (the product of $N$ Hilbert spaces) onto  $\bra{\mathbf {r}_1}\hdots\bra{\mathbf {r}_N} $ and  $\bra{\mathbf {k}_1}\hdots\bra{\mathbf {k}_N} $ coordinate systems.

Fields in QFT   are described by  the operators $\Psi(\mathbf {r},t)$, where $(\mathbf {r},t)$  is the space-time, and $\Phi(\mathbf {k},t)$ is the spatial Fourier transform of $\Psi(\mathbf {r},t)$.   A representation for a system of particles is based on  Fock states with occupation number  $\ket{n_{q_1},n_{q_2}, ,\hdots,n_{q_i},\hdots n_{q_K}}$,  where $n_{q_i}$ is the number of particles in a QM state $\ket{q_i}$.  The number of particles in a Fock state is then  $N=\sum_{i=1}^K n_{q_i}$, and a QFT state is described in a Fock space  as  $\ket{\state}=\sum_m \alpha_{m}  \ket{n_{q_1},n_{q_2}, ,\hdots,n_{q_i},\hdots}$, where $m$ is an index over configurations of a Fock state, $\alpha_{m} \in \mathbb{C}$, and $1=\sum_{m}| \alpha_{m}|^2$. The QFT  phase space state associated with any state $\ket{\state}$ is then given by $\left ( \bra{\state} \rho^{\QFT}_{\mathrm{r}}(\mathbf {r},t)\ket{\state},\,\bra{\state} \rho^{\QFT}_{\mathrm{k}}(\mathbf {k},t)\ket{\state}\right) $, where the probability density functions for the spatial coordinates are
\begin{align}
  \rho^{\QFT}_{\mathrm{r}}(\mathbf {r},t)&=\Psi^{\dagger}(\mathbf {r} ,t)\Psi(\mathbf {r},t)
  \\
  \rho^{\QFT}_{\mathrm{k}}(\mathbf {k},t)&=\Phi^{\dagger}(\mathbf {k},t)\Phi(\mathbf {k},t)\,.
  \label{eq:QFT-density-position}
\end{align}
The QFT coordinate-entropy  is then described by \eqref{eq:differential-entropy}, where we  dropped the superscript $\QFT$, as it will be clear which framework is used, QM or  QFT.

\subsection{Mixed Quantum States}

We now extend  the entropy \eqref{eq:differential-entropy} to mixed states. Consider a mixed state formed from $m\ge 2$ pure quantum states $\ket{\psi_j};j=1,\hdots, m$,  defined by the density matrix  $\rho^{\mathrm {M}}=\sum_{j=1}^m \lambda_j \ket{\psi_j}\bra{\psi_j}$, where  $\lambda_j >0$ and $1=\sum_{j=1}^m \lambda_j$. Then, projecting the density matrix onto the quantum coordinate phase space  basis yields  $\rho_{r}^{\mathrm {M}}(\mathbf{r},t)=\bra{\mathbf{r}}\rho^{\mathrm {M}}\ket{\mathbf{r}}=\sum_{j=1}^m \lambda_j \, |\psi_j(\mathbf{r},t)|^2$ and $\rho^{\mathrm {M}}_k(\mathbf{k},t)=\bra{\mathbf{k}}\rho^{\mathrm {M}}\ket{\mathbf{k}}=\sum_{j=1}^m \lambda_j \, |\phi_j(\mathbf{k},t)|^2$. These are the distributions associated with the phase space observables $\mathbf{r}$ and $\mathbf{p}=\hbar\mathbf{k} $.  We can also consider the distributions $\rho_j(\mathbf{r},\mathbf{k},t)=\lambda_j \, |\psi_j(\mathbf{r},t)|^2 \, |\phi_j(\mathbf{k},t)|^2 $, where  $1=\sum_{j=1}^m \kernBeforeIntegral\int \rho_j(\mathbf{r},\mathbf{k},t)\, \diff^3\mathbf{r}\,  \diff^3\mathbf{k}$,  which account for the quantum coordinate phase space as well as the probabilities associated with specifying the quantum state, namely the probabilities $\lambda_j; j=1,\hdots,m$. Thus, two different entropies can be considered, one quantifying just the randomness  of the phase space observables and the other also quantifying the randomness of specifying the quantum state.

When we are  quantifying just the randomness of the observables, then we must consider the probability densities in quantum coordinate phase space to be  $\rho_{r}^{\mathrm M}(\mathbf{r},t)$ and $\rho^{\mathrm {M}}_k(\mathbf{k},t)$, and the entropy can be evaluated by  the product of these distributions as follows
\begin{align}
  \entropyS^M
&=-\kernBeforeIntegral\int  \, \rho_{r}^M(\mathbf{r},t)  \rho^{\mathrm {M}}_k ( \mathbf{k},t)\,
 \, \ln \left ( \rho_{r}^M(\mathbf{r},t)  \rho^{\mathrm {M}}_k (\mathbf{k},t) \,  \right) \,
  \diff^3\mathbf{r}\, \diff^3\mathbf{k} \, .
\label{eq:differential-entropy-mixed-states-observables}
\end{align}

When  however, we attempt to quantify both sources of randomness,  the randomness associated with the observables and the randomness associated with the specification of the quantum state, then the density to be considered is  $\rho_j(\mathbf{r},\mathbf{k},t)=\lambda_j \, |\psi_j(\mathbf{r},t)|^2 \, |\phi_j(\mathbf{k},t)|^2 $, yielding the entropy
\begin{align}
  \entropyS^{M,\lambda_j}
&=-\sum_{j=1}^m \kernBeforeIntegral\int  \, \lambda_j \, |\psi_j(\mathbf{r},t)|^2  \, |\phi_{j}(\mathbf{k},t)|^2\,
 \, \ln \left ( \lambda_j \, |\psi_j(\mathbf{r},t)|^2    \, |\phi_j(\mathbf{k},t)|^2 \,  \right) \,
  \diff^3\mathbf{r}\, \diff^3\mathbf{k}
      \\
      &=-\sum_{j=1}^m \lambda_j \ln \lambda_j +\sum_{j=1}^m \lambda_j \, S_j
\,,
\label{eq:differential-entropy-mixed-states-hybrid}
\end{align}
where $S_j=-\kernBeforeIntegral\int  \, |\psi_j(\mathbf{r},t)|^2\,
 \, \ln |\psi_j(\mathbf{r},t)|^2\,
  \diff^3\mathbf{r}\, -\kernBeforeIntegral\int  \, |\phi_j(\mathbf{r},t)|^2\,
 \, \ln  |\phi_j(\mathbf{r},t)|^2 \,  \diff^3\mathbf{k}$ is the entropy of each pure state.

This entropy has two terms: the von Neumann entropy ($-\sum_{j=1}^m \lambda_j \ln \lambda_j $) and  the  average value of the entropies of the observables for each pure state  weighted by the mixed coefficients $\lambda_j$.   Clearly, the proposed entropy is larger than the von Neumann entropy.

Our proposed entropy for mixed states also  differs from Wehrl entropy because it is  a relativistic invariant based on a probability distribution of the observables,  and not on a quasiprobability distribution that lacks probability properties needed to characterize precisely the randomness of the observables.

One can interpret both entropies \eqref{eq:differential-entropy-mixed-states-observables} and \eqref{eq:differential-entropy-mixed-states-hybrid} as different generalizations to mixed states of the proposed entropy \eqref{eq:differential-entropy}. When the mixed state is reduced to one pure state, with $m=1$, both reduce to the entropy \eqref{eq:differential-entropy}.

In this paper we will focus on pure quantum states and leave as future research to extend further the study to mixed states.

\subsection{Spin-Entropy}
\label{sec:spin-DOF}
It is not possible to know simultaneously the spin of a particle in all three dimensional directions, and this uncertainty, or randomness, was exploited in the Stern–Gerlach experiment~\cite{gerlach1922experimentelle} to demonstrate the quantum nature of the spin. 
In order to identify the conjugate pair  of observables/operators associated to the spin we consider  the geometric quantization method\cite{geiger2021spin} and explore elsewhere~\cite{geiger2021spin} the entropy associated with the quantum spin phase space.

\section{Entropy Invariant Properties}
\label{sec:entropy}

\subsection{ Continuous Transformations of the Phase Space}
\label{subsec:coordinate-transformation}

In the QM setting, we investigate a point transformation of coordinates and  a translation in phase space of a  quantum reference frame \cite{aharonov1984quantum}.

Consider a  point transformation of position coordinates $F: \mathbf{r} \mapsto \mathbf{r}'$. It induces  the new conjugate momentum operator  \cite{dewitt1952point}
\begin{align}
    \label{eq:conjugate-momentum-DeWitt}
    \hat{\mathbf{p}}'=-\iu \hbar \left [\nabla_{\mathrm{r}'} + \frac{1}{2 }J^{-1}(\mathbf{r}')\nabla_{\mathrm{r}'}\cdot J(\mathbf{r}')\right]\, ,
\end{align}
where  $J(\mathbf{r}')\eqdefA J(F^{-1})(\mathbf{r}') = \frac{\partial \mathbf{r}(\mathbf{r}')}{\partial \mathbf{r}'}$ is the Jacobian  of $F^{-1}$ at $\mathbf{r}'$.

\begin{theorem}
 \label{proposition:change-coordinates-S}
The entropy  is invariant under a  point transformation of  coordinates.
\end{theorem}
\begin{proof}
Let $\entropyS$ be the entropy in phase-space relative to a conjugate Cartesian pair of coordinates $(\mathbf{r}, \mathbf{p})$. Let  $\mathbf{p}'$ be the momentum conjugate to $\mathbf{r}'$.
As the probabilities in infinitesimal volumes are invariant,
\begin{align}
    \label{eq:requirement}
    |\psi'(\mathbf{r}')|^2 \diff^3 \mathbf{r}' =|\psi(\mathbf{r}(\mathbf{r}' ) )|^2 \diff^3 \mathbf{r}(\mathbf{r}' )
  \text{ and }
    |\tilde \phi'(\mathbf{p}')|^2 \diff^3 \mathbf{p}'=|\tilde \phi(\mathbf{p}(\mathbf{p}') )|^2 \diff^3 \mathbf{p}(\mathbf{p}')\,,\quad
\end{align}
where $\mathbf{r}(\mathbf{r}' )= F^{-1}(\mathbf{r}')$ and $\mathbf{p}(\mathbf{p}' )= G^{-1}(\mathbf{p}')$ with $G: \mathbf{p} \mapsto \mathbf{p}'$ specified by~\eqref{eq:conjugate-momentum-DeWitt}.
Thus,  the probability density functions are $|\psi'(\mathbf{r}')|^2 $ and $|\tilde \phi'(\mathbf{p}')|^2$.   The Jacobian satisfies $\det J(\mathbf{r}')  \diff ^3 \mathbf{r}'=\diff ^3 \mathbf{r}(\mathbf{r}' )$, and applying this to \eqref{eq:requirement} we get $|\psi'(\mathbf{r}')|^2 \diff^3 \mathbf{r}'=|\psi(F^{-1}(\mathbf{r}'))|^2 \det J(\mathbf{r}') \, \diff ^3 \mathbf{r}'$, i.e., $|\psi'(\mathbf{r}')|^2=|\psi(F^{-1}(\mathbf{r}'))|^2 \det J(\mathbf{r}') $.
Similarly, we define $g(\mathbf{p}')=\det J(G^{-1})(\mathbf{p}')$ and so $g(\mathbf{p}') \diff ^3 \mathbf{p}'=\diff ^3 \mathbf{p}$, and to satisfy the infinitesimal probability invariance in momentum space
$  |\tilde \phi(\mathbf{p})|^2\, \diff ^3 \mathbf{p} =|\tilde \phi'(\mathbf{p}')|^2 \diff ^3 \mathbf{p}'$ at $\mathbf{p}'=G(\mathbf{p})$, we obtain $  |\tilde \phi'(\mathbf{p}')|^2 =|\tilde \phi(G^{-1}(\mathbf{p}')|^2 g(\mathbf{p}')$.

As noted in \cite{dewitt1952point}, there is an arbitrariness in the choice of $G$ that allows a new transformation $G'$  to be specified by~\eqref{eq:conjugate-momentum-DeWitt} with  $\det J(G'^{-1})(\mathbf{p}')=\frac{g(\mathbf{p}')}{f(\mathbf{p}')}$, i.e., the arbitrariness of $G'$ is equivalent to the choice of  a function $f(\mathbf{p}')$ to define the determinant of its (inverse) Jacobian. Then, associated with such a $G'$ we must also define a new  density function  $|\tilde \phi'(\mathbf{p}')|^2$ scaled by $\frac{1}{f(\mathbf{p}')}$,  producing an equally valid conjugate solution. Thus,
\begin{align}
    \entropyS_{\mathrm{r}}+\entropyS_{\mathrm{p}}   &= -\kernBeforeIntegral\int
 \diff^3\mathbf{r}\, \diff^3\mathbf{p}\, \rho_{\mathrm{r}} (\mathbf{r},t)  \rho_p (\mathbf{p},t) \ln \left ( \rho_{\mathrm{r}} (\mathbf{r},t)  \rho_p (\mathbf{p},t) \right)-3 \ln \hbar
 \\
 &= \entropyS_{\mathrm{r}'}+\entropyS_{\mathrm{p}'} - \langle \ln \det J^{-1}(\mathbf{r}')\rangle_{\rho'_{r'}} + \langle \ln  g(\mathbf{p}')
 \rangle_{\rho'_{p'}}
 \\
& = \entropyS_{\mathrm{r}'}+\entropyS_{\mathrm{p}'} \,,
\end{align}
and given the arbitrariness of $G$, we chose $g(\mathbf{p}')$  to satisfy
$  \langle\ln  g(\mathbf{p}')\rangle_{\rho'_{p'}}  = \langle\ln \det J^{-1}(\mathbf{r}')\rangle_{\rho'_{r'}}$.
\end{proof}

We next investigate translation transformations.  When a  quantum reference frame  is translated by $x_0$ along
$x$, the state $\ket{\psi_t}$ in the position representation becomes  $\psi(x-x_0,t) = \bra{x-x_0}\ket{\psi_t}= \bra{x}\hat T_P(-x_0)\ket{\psi_t}$, where  $\hat T_P(-x_0)=\eu^{\iu x_0 \, \frac{1}{\hbar}\hat{P}}$, and $\hat{P}$ is the momentum operator conjugate to $\hat{X}$.  When the reference frame is translated by $p_0$ along  $p$, the  state $\ket{\psi_t}$  in the momentum representation becomes  $\tilde \phi(p-p_0,t)= \bra{p-p_0}\ket{\psi_t}= \bra{p}\hat T_X(-p_0)\ket{\psi_t}$, where  $\hat T_X(-p_0)=\eu^{\iu \frac{1}{\hbar}p_0 \, \hat{X}}$, and $\hat X$ is the position operator conjugate to $\hat{P}$.

 \begin{theorem}[Frames of reference]
   \label{proposition:frame-invariance}
The entropy of a state is invariant under a change of a quantum reference frame by translations along $x$ and along $p$.
 \end{theorem}
 \begin{proof}
Let $\ket{\psi_t}$ be a state and  $\entropyS$ its entropy.
We start by showing  that $\entropyS_{x}= -\kernBeforeIntegral\int_{-\infty}^{\infty} \diff x \, |\psi(x,t)|^2 \ln |\psi(x,t)|^2$ is invariant under two types of translations:
\begin{enumerate}
    \item[(i)]  translations along $x$ by any $x_0$
  \begin{align}
    \label{eq:3}
    \entropyS_{x+x_0}&=-\kernBeforeIntegral\int_{-\infty}^{\infty} \diff x\,  |\psi(x+x_0,t)|^2 \ln |\psi(x+x_0,t)|^2
    =\entropyS_{x}\,,
  \end{align}
which is verified by changing variables under the infinite integration interval.
     \item[(ii)]   translations along $p$ by any $p_0$
 \begin{align}
   \psi_{p_0}(x,t) &=\bra{x}\hat T_X(p_0)\ket{\psi_t} = \int_{-\infty}^{\infty} \bra{x}\hat T_X(p_0)\ket{p}\bra{p}\ket{\psi_t} \diff p \\
   &=\int_{-\infty}^{\infty} \bra{x}\ket{p+p_0}\tilde{\phi}(p,t) \diff p
      =\int_{-\infty}^{\infty}  \frac{1}{\sqrt{2\piu}}\eu^{\iu \, x\, \frac{1}{\hbar}(p+p_0)}\tilde{\phi}(p,t) \diff p
      \\
      &
    = \psi(x,t) \,  \eu^{\iu \, x\, \frac{1}{\hbar}p_0}\,,
 \end{align}
 implying $|\psi_{p_0}(x,t)|^2 = |\psi(x,t)|^2$.
 \end{enumerate}
Similarly, by applying both translations to $\entropyS_{\mathrm{p}}=- \kernBeforeIntegral\int_{-\infty}^{\infty} \diff p \, |\tilde{\phi}(p,t)|^2 \ln |\tilde{\phi}(p,t)|^2$ we conclude that $\entropyS_{\mathrm{p}}$ is invariant under them too. Therefore $\entropyS=\entropyS_{x}+\entropyS_{\mathrm{p}}-3\ln \hbar$ is invariant under  translations in both $x$ and  $p$.
\end{proof}

\subsection{CPT Transformations}
\label{subsec:cpt}

We will be focusing on fermions, and thus on the Dirac spinors equation, though the results apply to bosons as well. The QFT Dirac Hamiltonian is
\begin{align}
    {\cal H}^{\mathscr{D}}=\int \diff^3 \mathbf{r} \, \Psi^{\dagger}(\mathbf{r},t) \left ( -\iu \planckbar\gamma^0 \vec{\gamma} \cdot \nabla+ m c \gamma^0  \right )\Psi(\mathbf{r},t)\,.
\end{align}
A QFT solution $\Psi(\mathbf{r},t)$ satisfies  $[{\cal H}^{\mathscr{D}},\Psi(\mathbf{r},t)]=-\iu \hbar \frac{\partial \Psi(\mathbf{r},t) }{\partial t}$ and the $C$, $P$, and  $T$ symmetries provide new solutions from  $\Psi(\mathbf{r},t)$.
   As usual,
  $\Psi^{\mathrm{C}}(\mathbf{r},t)= C\overline{\Psi}^{\tran}(\mathbf{r},t)$,
  $\Psi^\mathrm{P}(-\mathbf{r},t)= P\Psi(-\mathbf{r},t)$,
  $ \Psi^\mathrm{T}(\mathbf{r},-t)= T\Psi^*(\mathbf{r},-t)$,
and  $\psi^{\mathrm{CPT}}(-\mathbf{r},-t)= CPT\overline{\psi}^{\tran}(-\mathbf{r},-t)$.  For completeness, we briefly review the three operations, Charge Conjugation, Parity Change, and Time Reversal.

 Charge Conjugation  transforms particles $\Psi(\mathbf{r},t)$ into antiparticles  $\overline{\Psi}^\tran(\mathbf{r},t)= (\Psi^{\dagger}\gamma^0)^\tran(\mathbf{r},t)$.  As  $ C \gamma^{\mu } C^{-1}= -  \gamma^{\mu \tran}$,   $\Psi^{\mathrm{C}}(\mathbf{r},t)$ is  also  a solution for the same Hamiltonian. In the standard representation, $C=\iu \gamma^2\gamma^0$ up to a phase.
 Parity Change $P=\gamma^0$, up to a sign,  effects the transformation $\mathbf{r} \mapsto -\mathbf{r}$.
  Time Reversal  effects $t \mapsto -t$ and is carried by the operator $\mathscr{T}=T \hat K$, where $\hat K$ applies conjugation. In the standard representation $ T=\iu \gamma^1\gamma^3$, up to a phase.

\begin{theorem}[{Invariance of  the entropy under CPT-transformations}]
  \label{proposition:Entropy-CPT-invariant}
  Given a quantum field $\Psi(\mathbf{r},t)$, its Fourier transform $\Phi(\mathbf{k},t)$,  and its entropy $\entropyS_t$,
the entropies of $\, \Psi^{*}(\mathbf{r},t)$,   $\Psi^\mathrm{P}(-\mathbf{r},t)$, $\Psi^{\mathrm{C}}(\mathbf{r},t)$, $\Psi^\mathrm{T}(\mathbf{r},-t)$,  $\Psi^{\mathrm{CPT}}(-\mathbf{r},-t)$, and their corresponding Fourier transforms   are all equal to $\entropyS_t$.
\end{theorem}
\begin{proof}
  The probability densities of
  $\Psi^{*}(\mathbf{r},t)$,
  $\Psi^\mathrm{T}(\mathbf{r},-t)$, $\Psi^\mathrm{P}(-\mathbf{r},t)$,
  $\Psi^{\mathrm{C}}(\mathbf{r},t)$,
  and
  $\Psi^{\mathrm{CPT}}(-\mathbf{r},-t)$   are
 \begin{align}
   \rho^{*}_{\mathrm{r}}(\mathbf{r},t)&=\Psi^{\tran}(\mathbf {r},t) \Psi^*(\mathbf {r},t) =\Psi^{\dagger}(\mathbf {r},t)\Psi(\mathbf {r},t)=\rho(\mathbf{r},t)\,,
   \\
   \rho^{\mathrm{C}}_{\mathrm{r}}(\mathbf{r},t)&=\left(\overline{\Psi}^{\tran}\right)^{\dagger}(\mathbf{r},t)  \CC^{\dagger} \CC  \overline{\Psi}^{\tran}(\mathbf{r},t)   = \overline{\Psi}^{*}(\mathbf{r},t)\overline{\Psi}^{\tran}(\mathbf{r},t) = \rho_{\mathrm{r}}(\mathbf{r},t)\,,
   \\
   \rho^{\mathrm{P}}_{\mathrm{r}}(-\mathbf{r},t)&=\Psi^{\dagger}(\mathbf {r},t) (\gamma^0)^{\dagger} \gamma^0 \Psi(\mathbf {r},t)  = \Psi^{\dagger}(\mathbf {r},t)\Psi(\mathbf {r},t)= \rho_{\mathrm{r}}(\mathbf{r},t)\,,
   \\
   \rho^{\mathrm{T}}_{\mathrm{r}}(\mathbf{r},-t)&=\Psi^{\tran}(\mathbf {r},t) T^{\dagger} T \Psi^*(\mathbf {r},t)  = \Psi^{\tran}(\mathbf {r},t)\Psi^*(\mathbf {r},t)= \rho_{\mathrm{r}}(\mathbf{r},t)\,,
   \\
\rho^{\mathrm{CPT}}_{\mathrm{r}}(-\mathbf{r},-t)&=\left(\overline{\Psi}^{\tran}\right)^{\dagger}(\mathbf {r},t)  (\CC P T)^{\dagger} (\CC P T) \overline{\Psi}^{\tran}(\mathbf {r},t)   = \rho_{\mathrm{r}}(\mathbf{r},t)\,.
   \label{eq:CPT-densitities}
 \end{align}
 As  the densities are equal, so are the associated entropies.

Equations~\eqref{eq:CPT-densitities}  also hold for $\Phi(\mathbf{k},t)$ and its density.
Thus, both entropies terms in  $\entropyS_t=\entropyS_r+\entropyS_k$ are invariant under all CPT transformations.
\end{proof}

\subsection{Lorentz Transformations}
\label{subsec:Lorentz-transformation}

\begin{theorem}
\label{proposition:Lorentz-scalar}
  The entropy  is a {relativistic scalar.}
\end{theorem}
\begin{proof}
  The probability elements
    $\diff \pr(\mathbf{r},t) =\rho_{\mathrm{r}}(\mathbf{r},t)\diff^3\mathbf{r}$
  and
  $\diff \pr(\mathbf{k},t) =\rho_{\mathrm{k}}(\mathbf{k},t)\diff^3\mathbf{k}$
are invariant under Lorentz transformations because event probabilities  do not depend on the frame of reference. Consider a slice of the phase space with  frequency  $\omega_{\mathrm{k}}= \sqrt{ \mathbf{k}^2 c^2+ \left(\frac{m c^2}{\hbar}\right )^2}$. The volume elements $\frac{1}{\omega_{\mathrm{k}}} \diff^3\mathbf{k}$ and $\omega_{\mathrm{k}} \diff^3\mathbf{r}$,  are invariant under the Lorentz group~\cite{weinberg1995quantum1}, that is,
  $\frac{1}{\omega_{\mathrm{k}}} \diff^3\mathbf{k}=\frac{1}{\omega_{\mathrm{k}'}} \diff^3\mathbf{k}'$
  and
$\omega_{\mathrm{k}} \diff^3\mathbf{r}=\omega_{\mathrm{k}'} \diff^3\mathbf{r}'$,
implying
$     \diff V =  \diff^3\mathbf{k} \diff^3\mathbf{r} = \diff^3\mathbf{k}' \diff^3\mathbf{r}' =  \diff V '$,
where $\mathbf{r}'  $, $\mathbf{k}', $ and $ \omega_{\mathrm{k}'}$ result from applying  a Lorentz transformation  to  $\mathbf{r}$, $\mathbf{k}$, and $\omega_{\mathrm{k}}$. Thus, from the probability invariant elements we conclude that
 $\frac{1}{\omega_{\mathrm{k}}}\rho_{\mathrm{r}}(\mathbf{r},t)$ and $\omega_{\mathrm{k}} \rho_{\mathrm{k}}(\mathbf{k},t)$ are also
invariant under the  group. Thus, the phase space density $\rho_{\mathrm{r}}(\mathbf{r},t) \rho_{\mathrm{k}}(\mathbf{k},t)$ is an invariant under Lorentz transformations. Therefore the entropy  is  a relativistic scalar.
\end{proof}
Note that in QFT, one  scales the  operator  $\Phi(\mathbf{k},t)$ by $\sqrt{2\omega_{\mathrm{k}}}$, that is, one scales the creation and the annihilation operators ${\alpha}^{\dagger}(\mathbf{k})=\sqrt{\omega_{\mathrm{k}}}\, \mathbf{a}^{\dagger}(\mathbf{k})$ and ${\alpha}(\mathbf{k})=\sqrt{\omega}\,  \mathbf{a}(\mathbf{k})$. In this way,  the density operator $\Phi^{\dagger}(\mathbf{k},t)\Phi(\mathbf{k},t)$
scales with $\omega_{\mathrm{k}}$ and becomes  a relativistic scalar. Also, with such a scaling, the  infinitesimal  probability of finding a particle with momentum $\mathbf{p}=\hbar \mathbf{k}$ in the original reference frame is invariant  under the Lorentz transformation, though it would be found with  momentum $\mathbf{p}'=\hbar \mathbf{k}'$.

\section{The Minimum Entropy Value}
\label{subsec:minimum-entropy}

The third law of thermodynamics establishes $0$ as the minimum classical entropy. However, the minimum of the  quantum entropy must be positive  due to the uncertainty principle's lower bound.  Let $\uptheta(x)$ be $1$ for positive $x$ and $0$ elsewhere.
\begin{theorem}
 \label{proposition:lower-bound-S}
The minimum entropy of a particle  with spin $s$ is $3 (1+\ln \piu)+\uptheta(s)\ln 2\piu$.
\end{theorem}
\begin{proof}
  The entropy is the sum of the coordinate-entropy and the spin-entropy. The coordinate-entropy~\eqref{eq:differential-entropy} is $\entropyS_{\mathrm{r}}+ \entropyS_{\mathrm{k}}$.  Due to the entropic uncertainty principle
$\entropyS_{\mathrm{r}}+ \entropyS_{\mathrm{k}}\ge 3 \ln \ee \piu$ as shown in~\cite{hirschman1957note,beckner1975inequalities,  bialynicki1975uncertainty}, with
  $\entropyS_{\mathrm{k}}=\entropyS_{\mathrm{p}}-3\ln\hbar $.
To complete the proof, in~\cite{geiger2021spin}  we showed that the minimum spin-entropy  is $\uptheta(s)\ln 2\piu$.
\end{proof}
Higgs bosons in coherent states have the lowest possible entropy $3 (1+\ln \piu)$.

The dimensionless element of volume of integration to define the entropy  will not contain a particle unless  $\diff^3\mathbf{r}\, \diff^3\mathbf{k}\ge 1$, due to the  uncertainty principle, and this may be interpreted as a necessity of discretizing the phase space. We  note that  the  minimum entropy of the discretization of~\eqref{eq:differential-entropy}  is also $3 (1+\ln \piu)$, as shown in~\cite{dembo1991information}.

We point out that coherent states minimize   the uncertainty principle and as shown by Lieb~\cite{lieb2002proof} they also minimize Wehrl entropy. And as we show here,  they also minimize our proposed entropy.

\section{Entropy Evaluation for Some Physical Scenarios}
\subsection{The Coordinate-Entropy of Coherent States Increases With Time}
\label{sec:coherent-states}

 Coherent states, represented by state $\ket{\alpha}$,
are eigenstates of the annihilator operator. The 1D quantum phase space of observable variables $(x,p)$ can be constructed by the unitary  operator $U(x_0,p_0)=\eu^{\frac{\iu}{\hbar} (x_0 X -p_0 P) }$ applied to zero-state  $\ket{x=0,p=0}$, i.e., they can be constructed as  $\ket{\alpha}=\ket{x_0,p_0}=\eu^{\frac{\iu}{\hbar} (x_0 X -p_0 P)}\ket{0,0}$, where $\alpha=x_0+\iu p_0$. Projecting the state to  position space yields $\psi_{\alpha}(x)=\bra{x}\ket{\alpha}=\frac{\eu^{-\frac{p_0^2}{2}}}{\piu^{\frac{1}{4}}}\eu^{-\frac{1}{2}\left(x-\sqrt{2}\alpha\right)^2}$, where $\alpha=\frac{1}{\sqrt{2}}(x_0+\iu p_0)$. Squeeze states extend coherent states to all eigenstate solutions of the annihilator operator by allowing different variances to the Gaussian solution,  and together their representation in 3D position and momentum space are
\begin{align}
 \psi_{\mathrm{k}_0}(\mathbf {r}-\mathbf {r}_0)& =\frac{1}{2^3 \piu^{\frac{3}{2}}(\det \matrixsym{\Sigmabold})^{\frac{1}{2}}}\, \normalx{\mathbf {r}}{\mathbf {r}_0}{\matrixsym{\Sigmabold} }\, \eu^{\iu \mathbf {k}_0\cdot \mathbf {r}}\,,
 \\
 \Phi_{\mathrm{r}_0} (\mathbf {k}-\mathbf {k}_0) &=\frac{1}{2^3 \piu^{\frac{3}{2}}(\det \matrixsym{\Sigmabold}^{-1})^{\frac{1}{2}}}\, \normalx{\mathbf {k}}{\mathbf {k}_0}{\matrixsym{\Sigmabold}^{-1} }\, \eu^{\iu (\mathbf {k}-\mathbf {k}_0)\cdot \mathbf {r}_0}\,,
 \label{eq:coherent-state-3D}
\end{align}
where  $ \matrixsym{\Sigmabold} $ is the spatial covariance matrix. We will continue to refer to these states as coherent, with the understanding that we are including squeezed states, and that replacing the general covariance   $ \matrixsym{\Sigmabold} $ by $\mathrm{I}$  reduces to  the formal definition of coherent states.
The foundational material follows from most common textbooks, e.g., \cite{1968DiracLectures,weinberg1995quantum1,cohen2019quantum,Sakurai2017quantum}.

In \cite{GeigerKedem2021c} we proved that an initial coherent state~\eqref{eq:coherent-state-3D}  evolving according to the  energy
 $\hbar \omega(\mathrm{k})=\hbar \sqrt{ \mathbf{k}^2 c^2+ \left(\frac{m c^2}{\hbar}\right )^2}$, the entropy evolves as $ 3 (1+  \ln \piu )  + \frac{1}{2}\ln \det\big ( \Identitymatrix+ t^2 (\matrixsym{\Sigmabold}^{-1} \hessian)^2\big)$, where
 \begin{align}
\hessian_{ij}\eqdefA\hessian_{ij}(\mathbf{k}_0)& = \frac{\planckbar}{m} \left (1+\left (\frac{\planckbar \matrixsym{k}_0}{m c}\right)^2\right)^{-\frac{3}{2}}\left [ \deltau_{i,j}\left (1+\left (\frac{\planckbar \matrixsym{k}_0}{m c}\right)^2 \right) -\left (\frac{\planckbar \matrixsym{k_0}_i}{m c}\right) \left (\frac{\planckbar \matrixsym{k_0}_j}{m c}\right)\right ]\,,
\label{eq:Fourier-group-Hessian}
\end{align}
and $\hessian$ is positive definite. This suggests that quantum physics has an inherent dispersion mechanism to increase entropy for free fermion particles. Note that for coherent states of photons,  no dispersion occurs as  the electromagnetic Hamiltonian is non-dispersive.

\subsection{The Hydrogen Atom and Photon Emission}
\label{subsec:stab-hydr-ground}

The QED Hamiltonian  for the hydrogen atom is
\begin{align}
    H(p,r,q) & = \sum_{i=1}^3 \frac{\left( p^i-\frac{\ee}{c} A^i(q)\right )^2}{2 m} - \frac{\ee^2}{r}+\sum_{\lambda=1}^2  \planckbar \omega_{q} \, a^{\dagger }_{\lambda} (q)\, a_{\lambda}(q)\,,
\end{align}
where the  photon's helicity  $\lambda=1,2$,  $\omega_q=|q|c$,  the creation and the annihilation operators  of photons satisfy
 $[a_{\lambda}(p), a^{\dagger }_{\lambda'} (q)]=\deltau_{\lambda, \lambda'}\deltau(p-q)$, and the electromagnetic vector potential is
 \begin{align}
  \tilde A^i(q) & = \sqrt{2\piu  \planckbar c^2}\sum_{\lambda=1}^2 \frac{1}{\sqrt{\omega_q}} \left ( \epsilon^i_{\lambda}(q)\, a_{\lambda}(q)  + \epsilon^{*\iu }_{\lambda}(q) \, a^{\dagger}_{\lambda} (q)\right)\,,
 \end{align}
 and in  the Coulomb Gauge ($\nabla \cdot  A=0$), for $q=|q| ( \sin \theta_q \cos \phi_q, \sin \theta_q \sin  \phi_q, \cos \theta_q)$, the polarizations satisfy  $\epsilon_{1}(q) = ( \cos \theta_q \cos \phi_q, \cos \theta_q \sin  \phi_q, \sin \theta_q)$ and $\epsilon_{2}(q)  = (-\sin \phi_q, \cos \phi_q, 0)$.

 The state of the atom can be described by $\ket{n,l,m}_{{\mathrm e}^-}\ket{q,\lambda}_{\upgamma}$,
where  $n,l,m$ are the quantum numbers of the electron ${\mathrm e}^-$, and $q $ and $\lambda$ are the  momentum and the  helicity of the photon $\upgamma$.
 We next consider the  Lyman-alpha transition, $\ket{n=2,l=1,m=0}\ket{0} \, \rightarrow \, \ket{n=1,l=0,m=0}\ket{q,\lambda}$ with the emission of a photon with wavelength
 $\lambda\approx \SI{121.567d-9}{\metre}$.

We first evaluate the electron's entropy at both states $\ket{n=2,l=1,m=0}$ and $\ket{n=1,l=0,m=0}$. For simplicity, we consider the \Schroedinger approximation to describe the electron state with the energy change in this transition of $\Updelta E_{n=2\rightarrow n=1}\approx-\left(\frac{1}{2^2} - 1\right) \times \SI{13.6}{\electronvolt} = \SI{10.2}{\electronvolt}$. We now compute the difference between the final and the initial state entropy in three steps.
\begin{enumerate}
    \item[(i)]The position probability amplitudes described in \cite{bransden2003physics} and the associated entropies are
    \begin{align}
        \psi_{2, 1, 0}(\rho,\theta,\phi) &= \frac{1}{\sqrt{32\piu}} \left(\frac{1}{a_{0}}\right)^{\frac{3}{2}}\: \rho \ee^{-\frac{\rho}{2}} \cos(\theta)\: \rightarrow \,  \entropyS_{\mathrm{r}}(\psi_{2, 1, 0})\approx \num{6.120}+\ln \piu+3 \ln a_0\,,
        \\
        \psi_{1, 0,0}(\rho,\theta,\phi) &= \frac{1}{\sqrt{\piu}} \left(\frac{1}{a_{0}}\right)^{\frac{3}{2}}\: \ee^{-\rho} \:  \rightarrow \,  \entropyS_{\mathrm{r}}(\psi_{1, 0,0})\approx\num{3.000}+\ln \piu+3 \ln a_0 \,,
    \end{align}
     where $a_0 \approx \SI{5.292d-11}{\metre}$ is the Bohr radius, and $\rho={r}/{a_0}$.
    \item[(ii)] The momentum probability amplitudes described in \cite{bransden2003physics} and the associated entropies are
     \begin{align}
        \Phi_{2, 1, 0}(p, \theta_p, \phi_p)  &= \sqrt{\frac{128^2}{2\piu p_0^3}} \, \frac{p}{p_0} \, \left (1+\left(2\frac{p}{p_0}\right)^2\right )^{-3} \,\cos (\theta_p) \,,\\
        & \hspace*{-2em} \rightarrow \hspace{0.75em} \entropyS_p(\Phi_{2, 1, 0})\approx \num{0.042}+3 \ln p_0\,,
        \\
        \Phi_{1, 0, 0}(p, \theta_p, \phi_p) &= \sqrt{\frac{32}{\piu \, p_0^3}} \left ( 1+ \left(\frac{p}{p_0}\right)^2\right )^{-2} \,, \\
        & \hspace*{-2em}  \rightarrow \hspace{0.75em}  \entropyS_p(\Phi_{1, 0, 0})\approx \num{2.422}+3 \ln p_0\,  ,
    \end{align}
     where $p_0={\si{\planckbar}}/{a_0}$.
     \item[(iii)] Therefore, $\Updelta \entropyS_{2,1,0\rightarrow 1,0,0}=\entropyS_{\mathrm{r}}(\psi_{1, 0, 0})+ \entropyS_p(\Phi_{1, 0, 0}) -\entropyS_{\mathrm{r}}(\psi_{2, 1, 0})- \entropyS_p(\Phi_{2, 1, 0}) \approx \num{-0.740}\,.$
\end{enumerate}
Thus, the entropy of the electron is reduced by approximately $\num{0.740}$ during the transition $\ket{n=2,l=1,m=0} \, \rightarrow \, \ket{n=1,l=0,m=0}$.

We next evaluate the entropy associated with the randomness in the emission of the photon.  Due to energy conservation, the  energy must satisfy $|q| c\approx \SI{10.2}{\electronvolt}$, where $c$ is the speed of light. The associated energy uncertainty is very small. The main randomness for the photon is in specifying the direction of the emission.  The  angular momentum of the electron along $z$ ($m=0$) does not change between $\ket{n=2,l=1,m=0}$ and  $\ket{n=1,l=0,m=0}$.  The spin $1$ of the photon is along its motion, and conserves the total angular momentum of the system.  Thus, to conserve angular momentum along $z$,  the photon must be moving perpendicularly to the $z$ axis, that is, $\theta_q=\frac{\piu}{2}$, and so the polarization vectors  must be $   \epsilon_{1}(q) = ( 0, 0, 1) $ and $\epsilon_{2}(q)  = (-\sin \phi_q, \cos \phi_q, 0)$. The angle $\phi_q$ is completely unknown,  with the  entropy $\ln 2\piu$. Then we observe that the entropy increases, as
\begin{align}
     \Updelta \entropyS_{\ket{n=2,l=1,m=0}\ket{0} \, \rightarrow \, \ket{n=1,l=0,m=0}\ket{q,\lambda}} &\approx  \ln 2\piu - \num{0.740}
     =\num{1.098}\, .
\end{align}

Consider now  an apparent time-reversing scenario in which an apparatus emitted photons with energy $E_{\upgamma}=\hbar |\omega_{n=2,l=1,m=0}-\omega_{n=1, l=0, m=0}|$ to strike a hydrogen atom with its electron in the ground state. The photon had to follow a precise direction  towards the  atom, and  a very small uncertainty in the direction implies low photon entropy. Once the atom absorbs the photon,  the energy of the electron in the ground state suffices for a jump into an excited state. The entropy increases again, as the entropy of the excited state is larger than the entropy of the ground state (accounting for the low photon entropy).

\section{An Entropy Law and a Time Arrow}
\label{sec:entropy-law}

In classical statistical \mechanics, the entropy provides a time arrow through the second law of thermodynamics \cite{clausius1867mechanical}.  We have  shown that due to the dispersion property of the fermionic Hamiltonian, some states, such as  coherent states,  evolve with an increasing entropy.  However, current quantum physics is time reversible, and it is possible to  have states evolution  where the entropy oscillates. This includes the scenario in the Hydrogen atom studied earlier, where the excited state of the electron with no photon and the ground state of the electron with a photon emission are two possible states where quantum physics describe as an oscillation which we show elsewhere \cite{GeigerKedem2021c} leads to the entropy oscillation. 

We conjecture that isolated quantum systems (or states) can not have the total randomness associated with them decrease. 

\begin{law*}[The Entropy Law]
 \label{postulate:1}
 The entropy of an isolated quantum system is an increasing function of time.
\end{law*}
It is an information-theoretic conjecture about isolated quantum states, whereby information (the inverse of the entropy) can not be gained. 

An evidence for such a law is the Hydrogen atom scenario discussed earlier. According to QED, and due to photon fluctuations of the vacuum, the state of an electron in an excited state of the hydrogen atom  is in a superposition with the ground state, and we show elsewhere \cite{GeigerKedem2021c} the entropy would  decrease within a time interval ${2\piu}/{|\omega_{n=2,l=1,m=0}-\omega_{n=1, l=0, m=0}|}$. Instead, interrupting the oscillation, the electron jumps to the ground state and a photon is created/emitted, increasing the entropy. We conjecture that the entropy law is the trigger for the photon creation. 

Note that we also analyzed that this process is not reversible. Apparent reversibility when bombarding a hydrogen atom with photons targeted at it, creates photons with low entropy. Then the higher entropy solution is the absorption of the photon with the electron jumping to the excited state.

We  complete the paper wondering  if all  quantum states are always under the unitary evolution dictated by the Hamiltonian of the system as current QM asserts.  More precisely, in light of the conjectured entropy law, would all QM  evolution scenarios  produce an entropy that is always  increasing  (and by running it time backwards it would be  always decreasing) ? if so, the entropy law would simply discard the unitary evolution time direction where the entropy decreases. However, if under unitary evolution entropy oscillation scenarios can occur,  as described by the Fermi Golden Rule  transition \cite{fermi1950nuclear, dirac1927quantum}, then the entropy law would change our description of QM. The entropy law would trigger the collapse of a superposition of states to a state where the new evolution will have the entropy increasing.    The new initial state could represent the annihilation of particles and creation of new particles. Possibly, it  could describe  beta radiation, or the emission of the photon when the electron falls to the ground state in the Hydrogen atom, or the collision of two particles  to yield the creation of new particles.  In this case, like in the Copenhagen interpretation  of QM  the collapse of the state would occur, but different then the Copenhagen interpretation  it would   not require a measurement (nor an observer).  In fact a measurement could cause the collapse of the state due to its impact to the entropy evolution of the system (including the measurement apparatus).

\section{Conclusions}
\label{sec:conclusion}
Capturing all the information of a quantum state requires specifying the parameters associated with the DOFs of a quantum state as well as the intrinsic randomness of the quantum state. The intrinsic randomness is associated with a conjugate pair of observables, satisfying the uncertainty principle. 
We proposed a coordinate-entropy defined in the quantum phase spaces, the space of all possible states projected in the Fourier conjugate basis of position and spatial frequency.  Even though these observables are the same variables considered for defining the classical entropy, the  motivation and quantification is quite different. For the classical case, the randomness originates in the practical difficulties in specifying the DOFs precisely, while for the quantum pure state case the randomness is due to the intrinsic  quantum  state observables  characterized by a pair of conjugate observables that satisfy the uncertainty principle. 

This definition of the coordinate-entropy and quantum phase spaces possesses  desirable properties, including invariance in special relativity, and invariance  under CPT transformations.   We extended this entropy for the more general case where there is a randomness associated with specifying the quantum state, leading to a mixed quantum state. For mixed states, the entropy is always larger than von Neumann entropy due to the accounting of the randomness associated with the observables of each pure state.

We showed that due to the dispersion property of a fermionic Hamiltonian,    initial coherent states  will have the entropy increase over time.  We showed that the entropy increases when an electron in an excited state of the hydrogen atom falls to the ground state emitting a photon.  

We conjectured an entropy law that a closed quantum system must have the entropy to increase.  The motivation for the law is that information (inverse of the amount of randomness) can not increase in a closed quantum system. This law implies the irreversibility of time for scenarios where the entropy is not constant. 

We wonder if QM unitary evolution of states, with a dynamic given by the standard model (assuming it to be the correct one to describe quantum physics)  have the entropy to always increase (and by running the unitary evolution backwards, it would always decrease) so that oscillations do not occur. Then, the entropy law would simply discard the unitary evolution time direction where the entropy could be reduced.  Or, if  the unitary evolution leads to entropy oscillation scenarios, as it occurs for a two state system, and for the oscillation scenarios the entropy law  trigger the collapse of a superposition of states to a state where the new evolution will have the entropy increasing. Such collapse is accompanied by  particle creation or  annihilation.  In this case, the entropy law  determines that the event of particle creation and/or annihilation do occur, regardless of an observer performing a measurement.    Perhaps even the concept of measurement is a concept of a physical process that activates the conjectured entropy law so that particles are created and observed.  Thus, for example, the phenomena described by the double slit experiment would imply that at the sensors screen the absorption (annihilation) of the particle passing through the double slit occurs accompanied by the collapse of the particle state.  In this view, the QM superposition of states description accompanied by unitary evolution is completed with this entropy law, that effectively can cause the collapse of the state to a new state, where particles are created and annihilated.

\section{Acknowledgement} This paper is partially based upon work supported by both the National Science Foundation under Grant No.~DMS-1439786 and the Simons Foundation Institute Grant Award ID 507536 while the first author was in residence at the Institute for Computational and Experimental Research in Mathematics in Providence, RI, during the spring 2019 semester ``Computer Vision''  program.

 \bibliographystyle{abbrv}
 \bibliography{gk01}
\end{document}